\documentclass[10pt,oneside,twocolum,final]{IEEEtran}

\usepackage{algorithm}
\usepackage{algpseudocode}
\usepackage{amsfonts}
\usepackage{amsmath}
\usepackage{amssymb}
\usepackage{bm}
\usepackage{breqn}
\usepackage{caption}
\usepackage{cases}
\usepackage{cite}
\usepackage{color}
\usepackage{enumerate}
\usepackage{epsfig}
\usepackage{epstopdf}
\usepackage{exscale}
\usepackage{float}
\usepackage{graphics}
\usepackage{graphicx}
\usepackage{latexsym}
\usepackage{multicol}
\usepackage{multirow}
\usepackage{relsize}
\usepackage{stfloats}
\usepackage{subfigure}

\begin{document}

\title{ Multi-Pair Two-Way Massive MIMO DF Relaying Over Rician Fading Channels Under Imperfect CSI}
\IEEEoverridecommandlockouts
\author{\vspace{-0.1cm}
        Zhangjie~Peng, 
        Shuxian~Wang,
        Cunhua~Pan,~\IEEEmembership{Member,~IEEE,}
        Xianzhe~Chen,\\
        Julian~Cheng,~\IEEEmembership{Senior Member,~IEEE}
        and~Lajos~Hanzo,~\IEEEmembership{Fellow,~IEEE}
        \vspace{-0.2cm}

\thanks{This work was supported in part by the NSFC under Grant 61701307, and the open research fund of National Mobile Communications Research Laboratory, Southeast University under Grant 2018D14. The work of Lajos Hanzo was supported by the European Research Council’s Advanced Fellow Grant QuantCom (Grant No. 789028). (\emph{Corresponding author: Cunhua Pan}.)}
\thanks{Z. Peng is with the College of Information, Mechanical, and Electrical Engineering, Shanghai Normal University, Shanghai 200234, China, also with the National Mobile Communications Research Laboratory, Southeast University, Nanjing 210096, China, and also with the Shanghai Engineering Research Center of Intelligent Education and Bigdata, Shanghai Normal University, Shanghai 200234, China (e-mail: pengzhangjie@shnu.edu.cn).}
\thanks{S. Wang and X. Chen are with the College of Information, Mechanical and Electrical Engineering,
Shanghai Normal University, Shanghai 200234, China (e-mail: 278306849@qq.com; 1000479050@smail.shnu.edu.cn).}
\thanks{C. Pan is with the School of Electronic Engineering and Computer Science at Queen
Mary University of London, London E1 4NS, U.K. $( \text{e-mail: c.pan@qmul.ac.uk} )$.}
\thanks{J. Cheng is with the School of Engineering, The University of British
Columbia, Kelowna, BC V1V 1V7, Canada (e-mail: julian.cheng@ubc.ca).}
\thanks{L. Hanzo is with the School of Electronics and Computer Science,
University of Southampton, Southampton SO17 1BJ, U.K. $( \text{e-mail: lh@ecs.soton.ac.uk})$.}
\vspace{-1cm}

}

\maketitle

\newtheorem{lemma}{Lemma}
\newtheorem{proof}{Proof}
\newtheorem{theorem}{Theorem}
\newtheorem{remark}{Remark}
\newtheorem{proposition}{Proposition}

\vspace{-0.8cm}
\begin{abstract}
We investigate a multi-pair two-way decode-and-forward relaying aided massive multiple-input multiple-output antenna system under Rician fading channels, in which multiple pairs of users exchange information through a relay station having multiple antennas. Imperfect channel state information is considered in the context of maximum-ratio processing. Closed-form expressions are derived for approximating the sum spectral efficiency (SE) of the system. Moreover, we obtain the power-scaling laws at the users and the relay station to satisfy a certain SE requirement in three typical scenarios. Finally, simulations validate the accuracy of the derived results.


\end{abstract}

\vspace{-0.25cm}
\section{Introduction}
Driven by the dramatically increasing tele-traffic requirements, massive multiple-input multiple-output (MIMO) techniques have been extensively studied\cite{Marzetta2010Noncooperative},
where the users exchange their information via a base station (BS) equipped with hundreds of antennas.
Compared with traditional systems, massive MIMO systems substantially increase the spectral efficiency (SE) and energy efficiency (EE)
owing to their reduced transmit power\cite{jin2014ergodic}.
Hence, massive MIMO techniques play an important role in the current/next-generation networks.

The integration of massive MIMO applications with relay protocols can increase the network capacity and extend the coverage\cite{peng2013achievable}.
The power scaling laws of a one-way (OW) relay network were studied\cite{suraweera2013multi}.
However, OW relaying has the drawback of low SE
, which may potentially be doubled by two-way (TW) relaying protocols.
Explicitly, in TW relaying systems, multiple user pairs communicate with each other in a unique bidirectional channel\cite{feng2017power}.
Thus, mult-pair TW relaying systems have been developed to further improve the SE by adopting maximum-ratio (MR) processing at the relay.

Generally, two main relaying protocols are widely used: amplify-and-forward (AF) and decode-and-forward (DF).
DF relays decode the received signals before forwarding the re-encoded signals from a lower distance, which avoids interference and noise amplification\cite{gao2013sum}.
Additionally, DF TW relaying is capable of performing independent precoding and power allocation in each communication direction\cite{kong2018multipair}.
In practice, massive MIMO systems generally operate in line-of-sight (LOS) propagation conditions\cite{sayeed2011continuous}, and
Rician fading accurately models both LOS and diffuse scattered components \cite{zhang2014power,2018Multi,2019Power}.
Despite this, there are a paucity of analytical contributions under Rician fading channels for massive MIMO aided TW relaying system with imperfect CSI.

\vspace{-0.2cm}
\section{System Model}
We first study a multi-pair TW massive MIMO half-duplex DF relaying system that has $N$ pairs of users each employing a single antenna and an $M$-antenna relay (${T_R}$) under imperfect CSI.
The users at both ends are denoted by ${U_{A,i}}$ and ${U_{{B,i}}}$, for $i=1,...,N$.
Additionally, none of the communicating users has direct LOS links and can only exchange information through the TW relay.
The relay operates in time-division-duplex (TDD) mode.
We assume the reciprocity of the channels,
and denote the uplink (UL) and downlink (DL) channels between ${U_{{X,i}}}$ and ${T_R}$ by ${\mathbf{h}_{XR,i}}$ and $\mathbf{h}_{XR,i}^T$, respectively, where $X = A , B$ and $i = 1, \ldots ,N$.
Additionally, the channel matrix is formed as ${\mathbf{H}_{XR}} =\left[{{\mathbf{h}_{XR,1}},\;...,\;{\mathbf{h}_{XR,N}}} \right]$, $X = A , B$.
Then, the channel vector ${\mathbf{h}_{XR,i}}$ is expressed as $\mathbf{h}_{XR,i}^{} = \mathbf{g}_{XR,i}^{}\sqrt {\beta _{XR,i}^{}}$, where $\mathbf{g}_{XR,i}^{}$ represents the fast-fading element, while ${\beta _{XR,i}^{}}$ is the path-loss coefficient.
We assume that all the channels obey Rician distribution, and they are expressed as\cite{zhang2014power}
\setcounter{equation}{0}
\vspace{-0.25cm}
\begin{equation}\label{h_{XR,i}}\vspace{-0.25cm}
{\mathbf{g}_{XR,i}} \!\!=\!\! \sqrt {\!\!\frac{{{K_{XR,i}}}}{{{K_{XR,i}}\!\! +\!\! 1}}} {{\bar {\mathbf{g}}}_{XR,i}} \!\!+\!\! \sqrt {\!\!\frac{1}{{{K_{XR,i}} \!\!+\!\! 1}}} {{\tilde {\mathbf{g}}} _{XR,i}},   X\!\! \in\!\! \left\{\!{A,B}\! \right\},
\end{equation}
where ${{\bar {\mathbf{g}}}_{XR,i}}$ denotes the LOS part representing the deterministic component, ${{\tilde{\mathbf{g}}}_{XR,i}}$ denotes the scattered part representing the random component, and ${K_{XR,i}}$ is the Rician $K$-factor.
Perfect CSI is challenging to obtain for all the antennas.
We use the MMSE estimator at ${T_R}$ to estimate ${\mathbf{H}_{AR}}$ and ${\mathbf{H}_{BR}}$ \cite{zhang2014power},
where we have ${\mathbf{h}_{AR,i}} = {{\mathbf{\hat h}}_{AR,i}} + {\mathbf{e}_{AR,i}}$ and ${\mathbf{h}_{BR,i}} = {\mathbf{\hat h}_{BR,i}} + {\mathbf{e}_{BR,i}}$;
${\mathbf{\hat h}_{AR,i}}$ and ${\mathbf{\hat h}_{BR,i}}$ are the $i$th columns of the estimated matrices ${\mathbf{\hat H}_{AR}}$ and ${\mathbf{\hat H}_{BR}}$; ${\mathbf{e}_{AR,i}}$ and ${\mathbf{e}_{BR,i}}$ are the $i$th columns of the estimation error matrices ${\mathbf{E}_{AR}}$ and ${\mathbf{E}_{BR}}$, respectively.
${\mathbf{\hat H}_{XR}}$ and ${\mathbf{E}_{XR}}$ ($X\!=\!A$ or $B$) are independent.
Based on the assumption of the worst-case uncorrelated Gaussian noise, we can respectively obtain the variance of the estimation error vector elements ${\mathbf{e}_{AR,i}}$ and ${\mathbf{e}_{BR,i}}$ as $\sigma _{AR,i}^2 = \frac{{{\beta _{AR,i}}}}{{\left( {1 + \tau {p_p}{\beta _{AR,i}}} \right)\left( {{K_{AR,i}} + 1} \right)}}$ and $\sigma _{BR,i}^2 = \frac{{{\beta _{BR,i}}}}{{\left( {1 + \tau {p_p}{\beta _{BR,i}}} \right)\left( {{K_{BR,i}} + 1} \right)}}$, where $\tau $ denotes the channel training interval and $p_p$ is the transmit power of each pilot symbol.

The data transmission process is composed of two separate phases. First, all $N$ user pairs $\left( {{U_{A,i}},{U_{B,i}}} \right)$ simultaneously transmit their signals, $\left( {\sqrt {{p_{A,i}}} {x_{A,i}},\sqrt {{p_{B,i}}} {x_{B,i}}} \right)$, to ${T_R}$ in the UL phase. Thus, the UL signal received at ${T_R}$ is expressed as
\vspace{-0.25cm}
\begin{equation}\label{y_r}\vspace{-0.25cm}
{\mathbf{y}_r} = \sum\limits_{i = 1}^N {\left( {\sqrt {{p_{A,i}}} {\mathbf{h}_{AR,i}}{x_{A,i}} + \sqrt {{p_{B,i}}} {\mathbf{h}_{BR,i}}{x_{B,i}}} \right)}  + {\mathbf{n}_R},
\end{equation}
where ${{p_{A,i}}}$ and ${{p_{B,i}}}$ are the UL transmit powers of ${U_{A,i}}$ and ${U_{B,i}}$, respectively. The variables ${{x_{A,i}}}$ and ${{x_{B,i}}}$ respectively denote the signals transmitted by ${U_{A,i}}$ and ${U_{B,i}}$ with ${\mathbb{E}}\{ {{{| {{x_{A,i}}} |}^2}} \} = {\mathbb{E}}\{ {{{| {{x_{B,i}}} |}^2}} \} = 1$, where ${\mathbb{E}}\left\{  \cdot  \right\}$ represents the expectation operator. The vector ${\mathbf{n}_R} \sim {\cal C}{\cal N}\left( {{\bf{0}},{{\sigma _r^2\bf{I_N}}}} \right)$ is the additive white Gaussian noise (AWGN) at ${T_R}$.
The UL signal received at ${T_R}$ is decoded by multiplying it with the linear processing matrix ${{\bf{F}}_u}$, yielding
\vspace{-0.2cm}
\begin{equation}\label{Z_r}\vspace{-0.2cm}
{{{\bf{Z}}}_r}= {{\bf{F}}_u}{\mathbf{y}_r},
\end{equation}
where we have ${{\bf{F}}_u} = {\left[ {{\mathbf{\hat H}_{AR}},{\mathbf{\hat H}_{BR}}} \right]^H}$. From \eqref{y_r} and \eqref{Z_r}, we can derive the received signal of the $i$th pair of users after linear processing\cite{kong2018multipair}.

By contrast, in the DL phase, the signals received from all the users are decoded at the relay before transmission, while ${{\bf{F}}_d}$ is the linear precoding matrix which is applicable for the decoded signal $\bf{x}$. Therefore, the DL signal transmitted from the relay ${T_R}$ is given by
\vspace{-0.2cm}
\begin{equation}\label{y_t^{DF}}\vspace{-0.2cm}
\mathbf{y}_t = {\rho}{{\bf{F}}_d}\mathbf{x},
\end{equation}
where $\mathbf{x} = {\left[ {\mathbf{x}_A^T,\mathbf{x}_B^T} \right]^T}$, ${{\bf{F}}_d} = {\left[ {{\mathbf{\hat H}_{BR}},{\mathbf{\hat H}_{AR}}} \right]^*}$, and ${\rho}$ is adjusted for satisfying the transmit power constraint at the relay, i.e., $\mathbb{E}\left\{ {{{\left\| {\mathbf{y}_t} \right\|}^2}} \right\} = {p_r}$ and $\rho \!= \!\sqrt {\frac{{{p_r}}}{{\mathbb{E}\left\{ {\left\| {\bf{F}}_d  \right\|^2} \right\}}}} $.

Finally, the signals are forwarded to their respective destinations by the relay and the DL signal received at ${U_{X,i}}$ ($X\!=\!A$ or $B$) is given by
\vspace{-0.25cm}
\begin{equation}\label{z_{X,i}}\vspace{-0.2cm}
z_{X,i} = \mathbf{h}_{XR,i}^T\mathbf{y}_t + {\bf{n}}_{X,i},
\end{equation}
where ${\bf{n}}_{X,i} \sim {\cal C}{\cal N}\left( {0,\sigma _{X,i}^2} \right)$ is the AWGN at ${U_{X,i}}$.

\begin{figure*}[hb]
\setcounter{equation}{15}
\hrulefill
\vspace{-0.2cm}
\begin{equation}\label{gamma{RX,i}^{DF}}
{\rm{SINR}_{RX,i}^{}} = \frac{{\big| {\mathbb{E}\big\{ {\mathbf{h}_{XR,i}^T\mathbf{\hat h}_{XR,i}^*} \big\}} \big|_{}^2}}{{Var\big\{ {\mathbf{h}_{XR,i}^T\mathbf{\hat h}_{XR,i}^*} \big\} + Var\big\{ {\mathbf{h}_{XR,i}^T\mathbf{\hat h}_{\bar{X}R,i}^*} \big\} + \sum\limits_{j \ne i}^{} {\big( {\mathbb{E}\big\{ {\big| {\mathbf{h}_{XR,i}^T\mathbf{\hat h}_{BR,j}^*} \big|_{}^2} \big\} + \mathbb{E}\big\{ {\big| {\mathbf{h}_{XR,i}^T\mathbf{\hat h}_{AR,j}^*} \big|_{}^2} \big\}} \big) + \frac{1}{{\rho _{}^2}}} }}
\end{equation}
\vspace{-0.2cm}
\normalsize
\vspace{-0.2cm}
\end{figure*}

\begin{figure*}[hb]
\vspace{-0.2cm}
\setcounter{equation}{24}

\hrulefill
\vspace{-0.2cm}
\begin{align} \label{xi_{XR,ij}}
\xi_{XR,ij}={\frac{{{\beta _{XR,i}}{\beta _{AR,j}}}}{{\left( {{K_{XR,i}} + 1} \right)\left( {{K_{AR,j}} + 1} \right)}} \times \left( {\frac{{{K_{XR,i}} + {\eta _{XR,i}}}}{{1 + {\tau _p}{p_p}{\beta _{AR,j}}}} + {K_{AR,j}}{\eta _{XR,i}} + {K_{XR,i}}{\eta _{AR,j}} + {\eta _{XR,i}}{\eta _{AR,j}}} \right)}
\end{align}
\hrulefill
\vspace{-0.2cm}
\begin{align} \label{psi_{XR,ij}}
\chi_{XR,ij}={\frac{{{\beta _{XR,i}}{\beta _{BR,j}}}}{{\left( {{K_{XR,i}} + 1} \right)\left( {{K_{BR,j}} + 1} \right)}} \times \left( {\frac{{{K_{XR,i}} + {\eta _{XR,i}}}}{{1 + {\tau _p}{p_p}{\beta _{BR,j}}}} + {K_{BR,j}}{\eta _{XR,i}} + {K_{XR,i}}{\eta _{BR,j}} + {\eta _{XR,i}}{\eta _{BR,j}}} \right)}
\end{align}
\normalsize
\vspace{-0.2cm}
\hrulefill
\begin{align} \label{zeta_{XR,ij}}
&\zeta_{XR,ij}=\frac{{\beta _{AR,i}^{}\beta _{XR,j}^{}}}{{\left( {{K_{XR,j}} + 1} \right)\left( {{K_{AR,i}} + 1} \right)}}
\left(\!\! {\frac{{{K_{XR,j}}\! +\! {\eta _{XR,j}}}}{{1\! +\! {\tau _p}{p_p}\beta _{AR,i}^{}}}\! +\! {K_{AR,i}}{\eta _{XR,j}}\! +\! {K_{XR,j}}{\eta _{AR,i}} \!+ {\eta _{XR,j}}{\eta _{AR,i}}} \!\!\right)
\end{align}
\setcounter{equation}{33}
\normalsize
\vspace{-0.6cm}

\end{figure*}

\vspace{-0.1cm}
\section{Spectral Efficiency Analysis}

In this section, we investigate the SE of the TW half-duplex DF relaying system when imperfect CSI is considered at ${T_{R}}$.
The achievable sum SE of the system is defined as
\setcounter{equation}{5}
\vspace{-0.15cm}
\begin{equation}\label{R}
\vspace{-0.1cm}
R = \sum\limits_{i = 1}^N {R_i}
\end{equation}
where ${R_{i}}$ is the SE of the $i$th user pair, and it is defined as
\vspace{-0.15cm}
\begin{equation}\label{R-i}
\vspace{-0.05cm}
{R_i} = \min \left( {{R_{1,i}},{R_{2,i}}} \right)
\end{equation}
In \eqref{R-i}, ${R_{1,i}}$ is the SE of the $i$th user pair in the UL phase and
${R_{2,i}}$ is the SE of the $i$th user pair in the DL phase.
Without loss of generality, we will derive its closed-form approximations for the $i$th user pair.

\emph{1)}
As is in practical cases, the relay uses the estimated channel for signal detection.
Then, for the imperfect CSI case, ${R_{1,i}}$ is obtained as
\vspace{-0.25cm}
\begin{equation}\label{R_{1,i}^{IP}}
\vspace{-0.1cm}
{R_{1,i}}
 =  \lambda \mathbb{E} \left\{ {{{\log }_2}\left( {1 + \frac{{A_i^{} + B_i^{}}}{{C_i^{} + D_i^{} + E_i^{}}}} \right)} \right\},
\end{equation}
\!\!\!where we have $\lambda  = \frac{{T - \tau }}{{2T}}$, while  $A_i$ and $B_i$  represent the signals which ${U_{A,i}}$ and ${U_{B,i}}$ want to receive. Furthermore, $C_i$, $D_i$ and $E_i$ represent the estimation error, the inter-user interference and the compound noise,
respectively. The expressions of these five terms are given by
\vspace{-0.15cm}
\begin{align}
\vspace{-0.3cm}
 A_i^{} = {p_{A,i}}\left( {{{\big| {\mathbf{\hat h}_{AR,i}^H\mathbf{\hat h}_{AR,i}^{}} \big|}^2} + {{\big| {\mathbf{\hat h}_{BR,i}^H\mathbf{\hat h}_{AR,i}^{}} \big|}^2}} \right),\label{hat A_i^{IP}}
\end{align}
\vspace{-0.5cm}
\begin{align}
 B_i^{} = {p_{B,i}}\left( {{{\big| {\mathbf{\hat h}_{AR,i}^H\mathbf{\hat h}_{BR,i}^{}} \big|}^2} + {{\big| {\mathbf{\hat h}_{BR,i}^H\mathbf{\hat h}_{BR,i}^{}} \big|}^2}} \right),\label{hat B_i^{IP}}
\end{align}
\vspace{-0.5cm}
\begin{align}
 C_i^{} = {p_{A,i}}\left( {{{\big| {\mathbf{\hat h}_{AR,i}^H\mathbf{e}_{AR,i}^{}} \big|}^2} + {{\big| {\mathbf{\hat h}_{BR,i}^H\mathbf{e}_{AR,i}^{}} \big|}^2}} \right)\nonumber \\
 + {p_{B,i}}\left( {{{\big| {\mathbf{\hat h}_{AR,i}^H\mathbf{e}_{BR,i}^{}} \big|}^2} + {{\big| {\mathbf{\hat h}_{BR,i}^H\mathbf{e}_{BR,i}^{}} \big|}^2}} \right),\label{hat C_i^{IP}}
\end{align}
\vspace{-0.55cm}
\begin{align}
 D_i^{} = \sum\limits_{j \ne i}^{} {{p_{A,j}}} \left( {{{\big| {\mathbf{\hat h}_{AR,i}^H\mathbf{h}_{AR,j}^{}} \big|}^2} + {{\big| {\mathbf{\hat h}_{BR,i}^H\mathbf{h}_{AR,j}^{}} \big|}^2}} \right)\nonumber \\
+ \sum\limits_{j \ne i}^{} {{p_{B,j}}\left( {{{\big| {\mathbf{\hat h}_{AR,i}^H\mathbf{h}_{BR,j}^{}} \big|}^2} + {{\big| {\mathbf{\hat h}_{BR,i}^H\mathbf{h}_{BR,j}^{}} \big|}^2}} \right)} ,\label{hat D_i^{IP}}
\end{align}
\vspace{-0.55cm}
\begin{align}
\vspace{-0.4cm}
 E_i^{} = {\big\| {\mathbf{\hat h}_{AR,i}^{}} \big\|^2} + {\big\| {\mathbf{\hat h}_{BR,i}^{}} \big\|^2}.\label{hat E_i^{IP}}
\end{align}

Furthermore, the SE of the link ${U_{X,i}} \to {T_R}$ ($X\!=\!A$ or $B$) is given by
\vspace{-0.3cm}
\begin{equation}\label{hat R_{XR,i}^{DF}}
\vspace{-0.05cm}
R_{XR,i}^{} = \lambda \mathbb{E} \left\{ {{{\log }_2}\left( {1 + \frac{{X_i^{}}}{{C_i^{} + D_i^{} + E_i^{}}}} \right)} \right\}.
\end{equation}

In the DL phase, we can express the signals processed at ${U_{X,i}}$, ($X\!=\!A$ or $B$), after partial SIC according to (21) in \cite{kong2018multipair}.

The SE of the link ${T_R} \to {U_{X,i}}$ can be expressed as
\vspace{-0.1cm}
\begin{equation}\label{R_{RX,i}^{IP}}
\vspace{-0.05cm}
R_{RX,i}^{} = \lambda{\log _2}\left( {1 + \rm{SINR}_{RX,i}^{}} \right),
\end{equation}
where $X\!\! \in\!\! \left\{\!{A,B}\! \right\}$ and ${\rm{SINR}_{RX,i}^{}}$ defined in \eqref{gamma{RX,i}^{DF}} (at the bottom of the next page) is the corresponding SINR of ${U_{X,i}}$.
In \eqref{gamma{RX,i}^{DF}},  $\left\{ {\bar X} \right\} = \left\{ {A,B} \right\} \backslash  \left\{ {X} \right\}$.
\setcounter{equation}{16}
Thus, ${R_{2,i}}$ is defined as
\vspace{-0.1cm}
\begin{align}\label{R_{2,i}IP}
\vspace{-0.05cm}
{R_{2,i}}= \min \left( {{R_{AR,i}^{}},{R_{RB,i}^{}}} \right) + \min \left( {{R_{BR,i}^{}},{R_{RA,i}^{}}} \right).
\end{align}

\emph{2)}
In the following theorem, the closed-form approximation of $R_i^{}$ under imperfect CSI is formulated.
\begin{theorem}\label{theorem2}
In the imperfect CSI scenario, when the numberof the relay antennas, $M$, tends to infinity, the SE of the $i$th user pair employing MRC receivers is approximated as
\vspace{-0.15cm}
\begin{equation}\label{tilde R_i^{IP}}
\vspace{-0.05cm}
\tilde R_i^{} = \min \left( {\tilde R_{1,i}^{},\tilde R_{2,i}^{}} \right),
\end{equation}
where
\vspace{-0.3cm}
\begin{equation}\label{tilde hat R_{1,i}^{DF}}
\tilde R_{1,i}^{} = \lambda{\log _2}\Big( {1 + \frac{ M{p_{A,i}}\omega_{AR,i}^{2} + M{p_{B,i}}\omega_{BR,i}^{2} }{{\left( {\omega_{AR,i} + \omega_{BR,i} } \right){q_i} + Q_i}}} \Big),
\end{equation}
\vspace{-0.1cm}
\begin{equation}\label{tilde hat R_{2,i}^{DF}}
\tilde R_{2,i}^{} = \min \left( {\tilde R_{AR,i}^{},\tilde R_{RB,i}^{}} \right) + \min \left( {\tilde R_{BR,i}^{},\tilde R_{RA,i}^{}} \right),
\end{equation}
with
\vspace{-0.2cm}
\begin{equation}\label{tilde hat R_{AR,i}^{DF}}
\tilde R_{XR,i}^{} = \lambda{\log _2}\Big( {1 + \frac{ M{p_{X,i}}\omega_{XR,i}^{2} }{{\left( {\omega_{AR,i} + \omega_{BR,i}} \right){q_i} + Q_i}}} \Big),
\end{equation}
\begin{equation}\label{tilde hat R_{RA,i}^{DF}}
\tilde R_{RX,i}^{} = \lambda{\log _2}\Bigg( {1 + \frac{{M {p_r} {{ {\omega_{XR,i}^2} }}}}{{\sum\limits_{j = 1}^N {\left( { \omega_{AR,j} + \omega_{BR,j} + Z_{ij} } \right)} }}} \Bigg),
\end{equation}
\vspace{-0.5cm}
\begin{equation} \label{Q}
Q_i\!\! =\!\! \sum\limits_{j \ne i}^N( p_{A,j}(\xi_{AR,ij}+\xi_{BR,ij})
     \!+\!  p_{B,j}(\chi_{AR,ij}+\chi_{BR,ij})),
\end{equation}
\vspace{-0.5cm}
\begin{equation}\label{Z}
Z_{ij} = p_r (\zeta_{AR,ij}+\zeta_{BR,ij}),
\end{equation}
and $X\!\! \in\!\! \left\{{A,B} \right\}$,  ${\omega_{XR,i}^{}} =  {\frac{{\beta _{XR,i}^{}}}{{{K_{XR,i}} + 1}}({K_{XR,i}} + {\eta _{XR,i}})} $,  $\eta_{XR,i}= \frac{ {\tau _p}{p_p}{\beta _{XR,i}} }{ 1+ {\tau _p}{p_p}{\beta _{XR,i}}}$, ${q_i} = {p_{A,i}}\sigma _{AR,i}^2 + {p_{B,i}}\sigma _{BR,i}^2 + 1$. The terms $\xi_{AR,ij}$  and $\chi_{AR,ij}$ in \eqref{Q} are respectively defined by  \eqref{xi_{XR,ij}} and \eqref{psi_{XR,ij}} at the bottom of the next page, while $\zeta_{XR,ij}$ in \eqref{Z} is defined by \eqref{zeta_{XR,ij}} at the bottom of the next page.

\end{theorem}
\begin{proof}
See Appendix A.
\end{proof}

\emph{Theorem 1} presents the approximate expression of the SE  for the $i$th user pair under imperfect CSI. When ${\beta _{AR,i}^{}}$, ${\beta_{BR,i}}$, $K_{AR,i}$, $K_{BR,i}$, $\sigma_{A,i}$ and $\sigma_{B,i}$ are kept fixed, the SE is determined by the number of the user pairs $N$,  the number of the relay antennas $M$ and the transmit power $p_{A,i}$, $p_{B,i}$ and $p_{r}$. For fixed $p_{A,i}$, $p_{B,i}$, ${p_r}$ and ${p_p}$, we can see that $\tilde R_{1,i}^{}$ and $\tilde R_{2,i}^{}$ increase unboundedly in both the UL and DL phases as the number of relay antennas increases. Furthermore, in following Section IV, we will use \emph{Theorem 1} to investigate how the powers can be scaled down when the number of the relay antennas increases infinitely. It can be found that the sum SEs will converge to the upper limits for three typical cases, as  $M \to \infty $.
Additionally, the simulation in Section V verifies that the SE of the $i$th user pair increases with the Rician $K$-factor.

\vspace{-0.4cm}

\section{Power-Scaling Laws}
In this section, we quantify the power-scaling laws explicitly,
we analyze how the powers can be scaled down upon increasing $M$, while maintaining a certain SE.
Additionally, the transmit power of all users is set to be the same, i.e., ${p_{X,i}}$ = ${p_u}$, $X\!\! \in\!\! \left\{\!{A,B}\! \right\}$.

\setcounter{equation}{27}
We have ${p_u} = \frac{{{E_u}}}{{{M^\alpha }}}$, ${p_r} = \frac{{{E_r}}}{{{M^\varepsilon }}}$,  and ${p_p} = \frac{{{E_p}}}{{{M^\gamma }}}$, while ${{E_u}}$, ${{E_r}}$ and ${{E_p}}$ are all constants, $\alpha > 0$, $\varepsilon > 0$, and $\gamma > 0$.
Then,  it can be obtained from \emph{Theorem 1} that as $M \to \infty $, we have $\eta_{XR,i} \xrightarrow{} 0$, $Q_i \xrightarrow{} 0$, $Z_{ij} \xrightarrow{} 0$,  $q_i\xrightarrow{} 1$. Thus, as $M \to \infty $, $\tilde R_i$  defined by \eqref{tilde R_i^{IP}} in \emph{Theorem 1} converges according to
\begin{equation}\label{R_is3^{DF}}
\tilde R_i \xrightarrow{M \to \infty }  \min \left( {\bar R_{1,i}^{},\bar R_{2,i}^{}} \right) \triangleq \bar{R}_i,
\end{equation}
where
\vspace{-0.4cm}
\begin{equation}\label{bar Rs3_{1,i}^{DF}}
\bar R_{1,i}^{}=\lambda{\log _2}\Big( {1 + \!\!\frac{{E_u^{}}}{{M_{}^{\alpha  - 1}}}\frac{{\psi_{AR,i}^2+\psi_{BR,i}^2}}{{ \psi_{AR,i}+\psi_{BR,i} }}} \Big),
\end{equation}
\vspace{-0.3cm}
\begin{equation}\label{bar Rs3_{2,i}^{DF}}
\vspace{-0.05cm}
\bar R_{2,i}^{} = \min \left( {\bar R_{AR,i}^{},\bar R_{RB,i}^{}} \right) + \min \left( {\bar R_{BR,i}^{},\bar R_{RA,i}^{}} \right),
\end{equation}
with
\vspace{-0.4cm}
\begin{equation}
\bar R_{XR,i}^{}=\lambda{\log _2}\Big( {1 + \!\!\frac{{E_u^{}}}{{M_{}^{\alpha  - 1}}}\frac{{\psi_{XR,i}^2}}{{ \psi_{AR,i}+\psi_{BR,i} }}} \Big),
\end{equation}
\vspace{-0.5cm}
\begin{equation}
\bar R_{RX,i}^{} = \lambda{\log _2}\Bigg( {\!\!1 +\!\!\frac{{E_r^{}}}{{M_{}^{\varepsilon  - 1}}} \frac{{{\psi_{XR,i}^2}}}{{\sum\limits_{j = 1}^N {\left( \psi_{AR,i}+\psi_{BR,i} \right)} }}} \Bigg),\label{bar Rs3_{RA,i}^{DF}}
\end{equation}
\vspace{-0.4cm}
\begin{equation}
{\psi_{XR,i}^{}} =  \frac{{\beta _{XR,i}^{}}{K_{XR,i}}}{{{K_{XR,i}} + 1}} .
\vspace{-0.2cm}
\end{equation}
We observe that the asymptotic SEs of the $i$th user pair are closely related to the values of $\alpha$ and $\varepsilon$. Moreover, we find that $R_i^{}$ is independent of $\gamma $ when $M$ becomes large. Next, we analyze the effect of $\alpha$ and $\varepsilon$ on the SE.
\begin{itemize}
\item When $\alpha > 1$ or $\varepsilon > 1$, $R_i^{}$ converge to zero. When ${p_u}$ or ${p_r}$ is reduced excessively, the asymptotic SEs will tend to zero.
\item When $0 < \alpha  < 1$ and $0 < \varepsilon  < 1$, $R_i^{}$ grow without limit. This suggests that ${p_u}$ and ${p_r}$ can be reduced more drastically to obtain fixed SEs.

\item When $\alpha = 1$, $0 < \varepsilon  \le 1$ or $\varepsilon = 1$, $0 < \alpha \le 1$, $R_i^{}$ converge to a positive limit. We now study how much ${p_u}$ or ${p_r}$ or both can be scaled down, while maintaining a certain SE. We focus on three cases: 1) Case I: $\alpha  = \varepsilon  = 1$; 2) Case II: $\alpha  = 1$, and $ 0 < \varepsilon  < 1$; 3) Case III: $ 0 < \alpha  < 1$ and $\varepsilon  = 1$.
\end{itemize}

\vspace{-0.4cm}
\subsection{\emph{\bf{Case I:}} $\alpha  = \varepsilon  = 1$.}

For fixed ${{E_u}}$, ${{E_r}}$ and ${E_p}$, by substituting $\alpha  = 1$  and $ \varepsilon  = 1$  into \eqref{bar Rs3_{1,i}^{DF}} - \eqref{bar Rs3_{RA,i}^{DF}},
as $M \to \infty$, we can simplify the SE in \eqref{bar Rs3_{1,i}^{DF}} - \eqref{bar Rs3_{RA,i}^{DF}} as
\vspace{-0.4cm}
\begin{align}\label{bar Rc4_{1,i}^{DF}}
\bar R_{1,i}^{} = \lambda{\log _2}\bigg( {1 + \frac{{E_u^{}\left( {{{ {\psi_{AR,i}^2} }} + {{ {\psi_{BR,i}^2} }}} \right)}}{{ {\psi_{AR,i} + \psi_{BR,i}} }}} \bigg),
\end{align}
\vspace{-0.3cm}
\begin{equation}\label{bar Rc4_{2,i}^{DF}}
\bar R_{2,i}^{} = \min \left( {\bar R_{AR,i}^{},\bar R_{RB,i}^{}} \right) + \min \left( {\bar R_{BR,i}^{},\bar R_{RA,i}^{}} \right),
\end{equation}
with
\vspace{-0.2cm}
\begin{equation}
\bar R_{XR,i}^{}=\lambda{\log _2}\Big( {1 + \!\!\frac{E_u {\psi_{XR,i}^2}}{{ \psi_{AR,i}+\psi_{BR,i} }}} \Big),
\end{equation}
\vspace{-0.3cm}
\begin{equation}\label{bar Rc4_{RA,i}^{DF}}
\bar R_{RX,i}^{} = \lambda {\log _2}\Bigg( {1 + \frac{{{E_r}{{ \psi_{XR,i}^2 }}}}{{\sum\limits_{j = 1}^N {\left( {\psi_{AR,j} + \psi_{BR,j}} \right)} }}} \Bigg).
\end{equation}

Based on \eqref{bar Rc4_{1,i}^{DF}} - \eqref{bar Rc4_{RA,i}^{DF}},
the limit of $R_i^{}$ also increases with ${{E_u}}$ and ${{E_r}}$, and decreases with $N$.

\vspace{-0.4cm}
\subsection{\emph{\bf{Case II:}} $\alpha  = 1$, and $ 0 < \varepsilon  < 1$.}
\begin{figure*}[hb]
\vspace{-0.2cm}
\setcounter{equation}{37}

\hrulefill
\vspace{-0.1cm}
\begin{align}\label{bar Rc6_{1,i}^{DF}}
\tilde R_i \to&  \bar{R}_{2,i} = \lambda{\log _2}\Big( {1 + {{E_r^{}{{ {\psi_{AR,i}^2} }}}}\Big/{{\sum\limits_{j = 1}^N {\left( {\psi_{AR,j} + \psi_{BR,j}} \right)} }}} \Big) +  \lambda{\log _2}\Big( {1 + {{E_r^{}{{ {\psi_{BR,i}^2} }}}} \Big/ {{\sum\limits_{j = 1}^N {\left( {\psi_{AR,j} + \psi_{BR,j}} \right)} }}} \Big)
\vspace{-0.15cm}
\end{align}

\normalsize
\vspace{-0.9cm}
\end{figure*}

For fixed ${{E_u}}$, ${{E_r}}$ and ${E_p}$, substituting $\alpha  = 1$ and $0 < \varepsilon  < 1$ into \eqref{R_is3^{DF}} - \eqref{bar Rs3_{RA,i}^{DF}}, we can obtain that $\tilde R_i$ converges to $\bar R_{1,i}$ given by \eqref{bar Rc4_{1,i}^{DF}}, i.e., $\tilde R_i^{} \to \bar R_{1,i}^{} $,  as $M \to \infty$.

We find that the asymptotic SE of $R_i^{}$ is decided by the UL phase, which increases with ${{E_u}}$, while ${{E_r}}$ has no effect on the asymptotic SE. Furthermore, as $M$ increases, the SE of each user will become lower in the UL than that of the DL phase.

\vspace{-0.3cm}
\subsection{\emph{\bf{Case III:}} $ 0 < \alpha  < 1$ and $\varepsilon  = 1$.}
For fixed ${{E_u}}$, ${{E_r}}$ and ${E_p}$,  by substituting $0 < \alpha  < 1$ and $\varepsilon  = 1$ into \eqref{R_is3^{DF}} - \eqref{bar Rs3_{RA,i}^{DF}}, $\tilde R_i$ converges to $\bar R_{2,i}$ given by \eqref{bar Rc6_{1,i}^{DF}} (at the bottom of the next page), i.e., $\tilde R_i^{} \to \bar R_{2,i}^{}$,  when $M \to \infty$.

The asymptotic SE of $R_i^{}$ increases with ${{E_r}}$, while ${{E_u}}$ has no effect on the asymptotic SE. Furthermore, as $M$ tends to infinity, the SE of each user in the DL phase will be lower than that of the UL phase.

\begin{figure}[t]
\vspace{-0.9cm}
\centering
\includegraphics[scale=0.52]{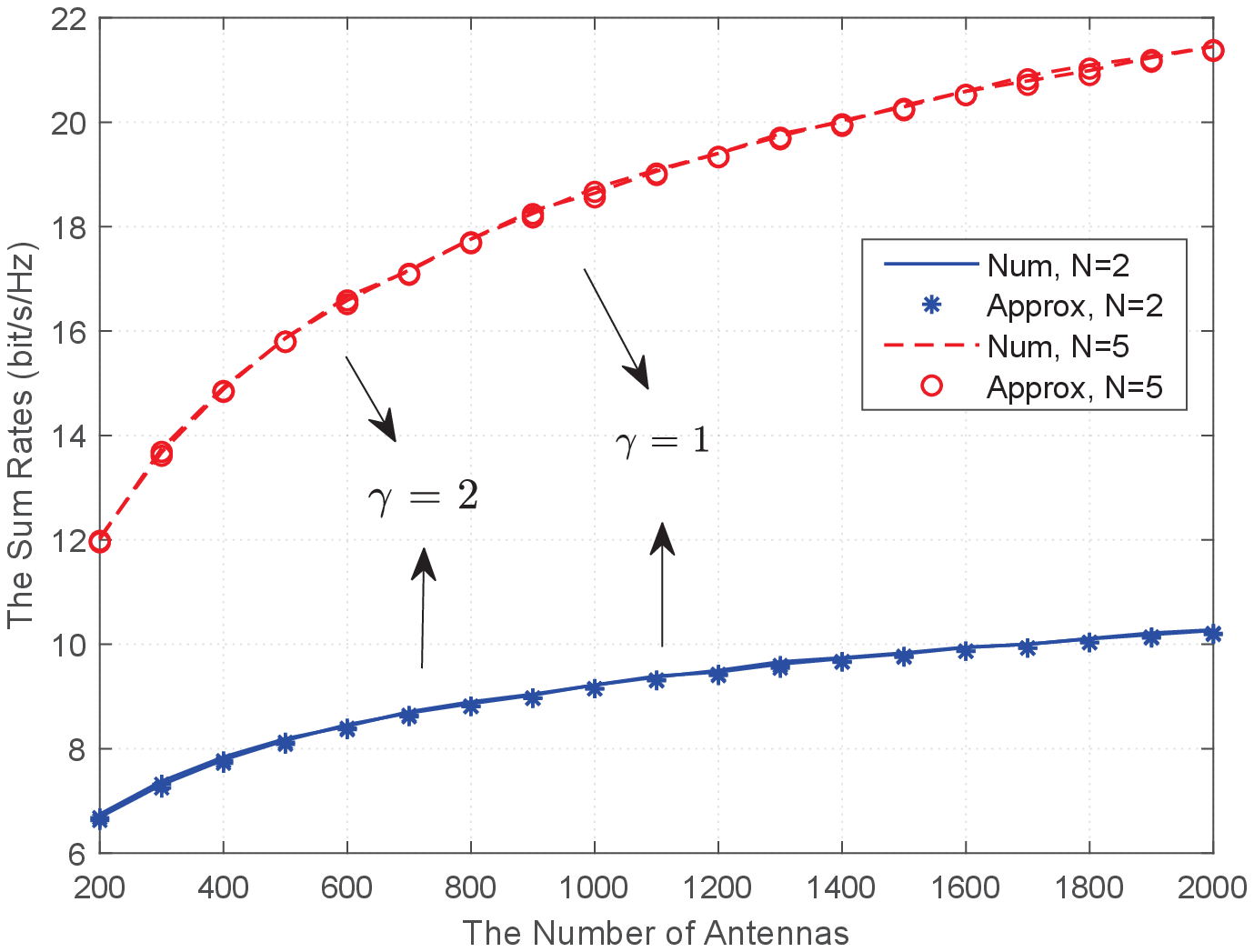}
\vspace{-0.2cm}
\caption{Sum SEs versus $M$ for ${p_u} = {E_u}$, ${p_r} = {E_r}$, and ${p_p} = \frac{{{E_p}}}{{{M^\gamma }}}$.}
\label{figure_pp}
\vspace{-0.2cm}
\end{figure}

\begin{figure}[t]
\centering
\includegraphics[scale=0.52]{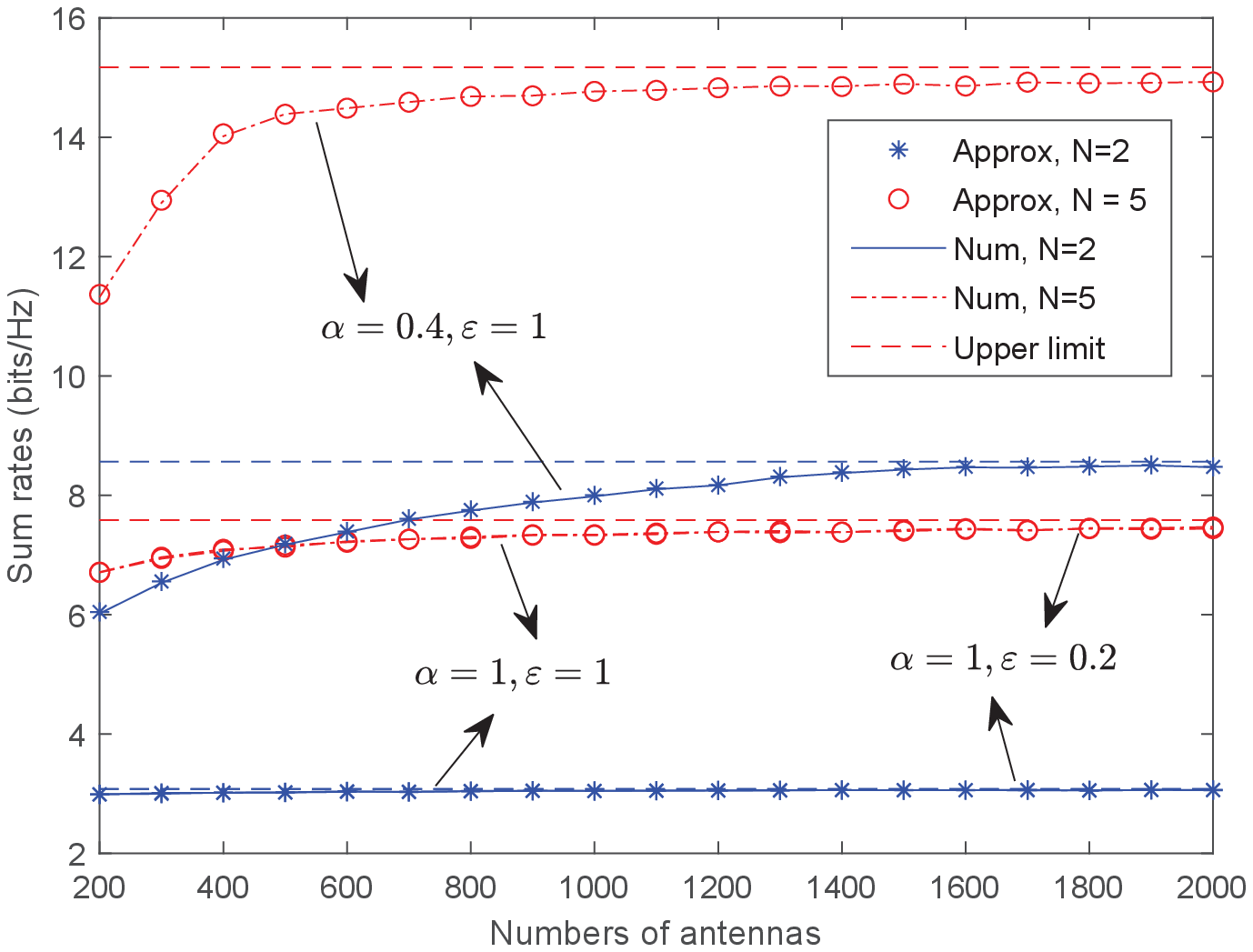}
\vspace{-0.2cm}
\caption{Sum SEs versus $M$ for ${p_u} = \frac{{{E_u}}}{{{M^\alpha }}}$, ${p_r} = \frac{{{E_r}}}{{{M^\varepsilon }}}$, ${p_p} = \frac{{{E_p}}}{{{M^\gamma }}}$, ${\alpha  \le 1}$, and ${\varepsilon  \le {\rm{1}}}$.}
\label{figa}
\vspace{-0.6cm}
\end{figure}

\vspace{-0.4cm}

\section{Numerical Results}
In this section, we verify the main results of this letter by numerical results.
The length of the coherence time is set to $T$ = 196 symbols.
For simplicity, we assume that ${E_u} = {E_p} = 10$ dB, ${E_r} = 20$ dB, ${\beta _{AR,i}} = {\beta _{BR,i}} = 1$. Furthermore, each Rician K-factor is set to the same value.

Numerical results are provided for the sum SE with ${K_{XR,i}} = 5$ dB for all $i$ and  $N = 2$ or $5$ in Fig. \ref{figure_pp} - Fig. \ref{fig:subfig:b}.
In Fig. \ref{figure_pp}, the exact expressions and the approximations are compared. It is observed that the pairs of curves match well for $N = 2$ and $5$.
The sum SE increases with $M$, as expected. Furthermore, when $N = 5$, the sum SE is almost twice as high as that for $N =2$. This indicates that the sum SE increases with $N$. Additionally, when ${p_u}$ and ${p_r}$ are unchanged, the transmit power ${p_p}$ of the pilot symbol is cut down, we find that the sum SE is independent of the choice of $\gamma $, when $M$ becomes large.

In Fig. \ref{figa} and Fig. \ref{fig:subfig:b},  the corresponding sum SEs and upper limits are presented when ${p_u}$ and ${p_r}$ are scaled down.  Explicitly, Fig. \ref{figa} investigates three cases using different settings of $\alpha$ and $\varepsilon$. In line with Case I-III of Section IV,
the sum SEs saturate as $M$ tends to infinity in all three circumstances. We observe that $\alpha =1 $, $\varepsilon = 1$  and $\alpha =1 $, $\varepsilon = 0.2$ achieve the same sum SE, because it is decided by the UL phase.
Fig. \ref{figa} also illustrates the corresponding upper limits  $\bar{R}_i$ defined in \eqref{R_is3^{DF}} for the sum SEs. It can be observed that the sum SEs converge to the corresponding upper limits with the increasing $M$ for these three different cases. Fig. \ref{fig:subfig:b} studies three different scenarios, i.e., 1) $\alpha  > 1$, and $\varepsilon  > 0$, 2) $\alpha  > 0$, and $\varepsilon > 1$, 3) $ \alpha > 1$, and $\varepsilon > 1$.
As expected, the sum SEs converge to zero as $M$ grows. The reduction of the sum SEs is faster for larger scaling parameters for different $N$.

Fig. \ref{figure_K} depicts the sum SE versus the Rician $K$-factor. Herein, we set $N=5$, ${p_u} = \frac{{{E_u}}}{M}$, ${p_r} = \frac{{{E_r}}}{M}$, and ${p_p} = \frac{{{E_p}}}{M}$.
We compare the sum SEs when ${K_{XR,i}} =3, 5, 10$ dB.
It is clear that the approximations match well with the exact expressions in all scenarios. The sum SE increases with $M$, as expected.
As the Rician $K$-factor grows, the sum SE increases.

\begin{figure}[t]
\vspace{-1.1cm}
\centering
\includegraphics[scale=0.42]{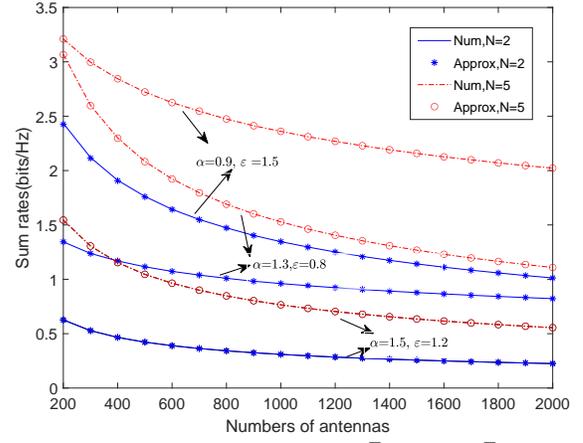}
\vspace{-0.2cm}
\caption{Sum SEs versus $M$ for ${p_u} = \frac{{{E_u}}}{{{M^\alpha }}}$, ${p_r} = \frac{{{E_r}}}{{{M^\varepsilon }}}$, ${p_p} = \frac{{{E_p}}}{{{M^\gamma }}}$, $\alpha  > 1$, or $\varepsilon  > 1$.}
\label{fig:subfig:b}
\vspace{-0.3cm}
\end{figure}

\vspace{-0.7cm}
\begin{figure} [t]
\vspace{-0.1cm}
\centering
\includegraphics[scale=0.52]{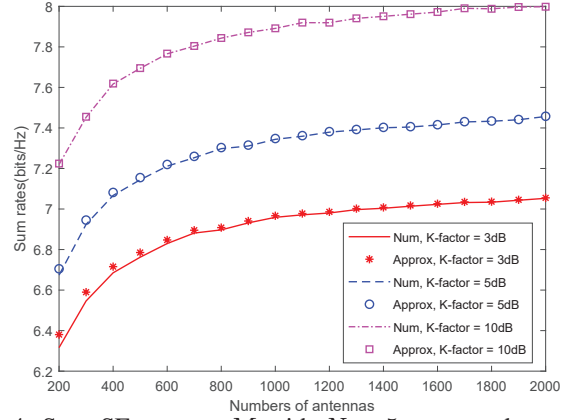}
\vspace{-0.3cm}
\caption{Sum SEs versus $M$ with $N=5$ users and ${p_u} = \frac{{{E_u}}}{M}$, ${p_r} = \frac{{{E_r}}}{M}$, and ${p_p} = \frac{{{E_p}}}{{{M}}}$.}
\label{figure_K}
\vspace{-0.7cm}
\end{figure}

\section{Conclusions}
We studied a multi-pair TW DF relay system using a massive MIMO scheme at the relay, upon adopting a MR receiver.
Furthermore we derived the exact expressions and approximations of the SE over Rician fading channels for an imperfect CSI scenario.
Finally, we quantified the trade-off between the SE, ${p_u}$ and ${p_r}$.
Additionally, the sum SE of the imperfect CSI scenario increases, as the Rician $K$-factor grows.

\begin{appendices}
\vspace{-0.5cm}

\section{Proof of Theorem \ref{theorem2} }
\setcounter{equation}{38}

From \eqref{R_{1,i}^{IP}}, by using Lemma 1 of \cite{zhang2014power}, $R_{1,i}^{}$ in \eqref{R_{1,i}^{IP}} can be approximated as
\begin{equation}\label{RIP}
R_{1,i}^{} \!\!\approx \lambda  {\log _2}\left(\! {1\!\! +\!\! \frac{{\mathbb{E}\left\{ {A_i^{}} \right\} + \mathbb{E}\left\{ {B_i^{}} \right\}}}{{\mathbb{E}\left\{ {C_i^{}} \right\} + \mathbb{E}\left\{ {D_i^{}} \right\} + \mathbb{E}\left\{ {E_i^{}} \right\}}}} \!\!\right)\buildrel \Delta \over = \tilde R_{1,i}^{}.
\end{equation}
Then, we will calculate the terms $\mathbb{E}\left\{ {A_i^{}} \right\}$, $\mathbb{E}\left\{ {B_i^{}} \right\}$, $\mathbb{E}\left\{ {C_i^{}} \right\}$, $\mathbb{E}\left\{ {D_i^{}} \right\}$ and $\mathbb{E}\left\{ {E_i^{}} \right\}$ .

By using  Lemma 5 of \cite{zhang2014power} and retaining the dominant components, we can approximate $\mathbb{E}\left\{ { A_i^{}} \right\}$, $\mathbb{E}\left\{ { B_i^{}} \right\}$ and $\mathbb{E}\left\{ { C_i^{}} \right\}$ respectively as
\vspace{-0.2cm}
\begin{equation}\label{EAj}
\vspace{-0.1cm}
\mathbb{E}\left\{ { A_i^{}} \right\}  \approx M^2 \omega_{AR,i}^2, \;\;\;\; \mathbb{E}\left\{ { B_i^{}} \right\}  \approx M^2 \omega_{BR,i}^2.
\end{equation}
\begin{equation}\label{ECj}
\vspace{-0.2cm}
\mathbb{E}\left\{ { C_i^{}} \right\}
\!=\!\! M\left( {\omega_{AR,i} \!+\! \omega_{BR,i} } \right)\left( {{p_{A,i}} \sigma _{AR,i}^2 \!+\! {p_{B,i}}\sigma _{BR,i}^2} \right). \!
\end{equation}

From \eqref{hat D_i^{IP}}, $\mathbb{E}\left\{ { D_i^{}} \right\}$ can be written as
\vspace{-0.1cm}
\begin{align}\label{EDjIP}
\vspace{-0.2cm}
\mathbb{E}\left\{ { D_i^{}} \right\} \!\!
 &= \!\!\sum\limits_{j \ne i}^{}\! {{p_{A,j}}\!\Big(\! {\mathbb{E}\Big(\! {{{\big| {\mathbf{\hat h}_{AR,i}^H\mathbf{h}_{AR,j}^{}} \big|}^2}} \!\Big) \!+ \!\mathbb{E}\Big(\! {{{\big| {\mathbf{\hat h}_{BR,i}^H\mathbf{h}_{AR,j}^{}} \big|}^2}} \!\Big)}\! \Big)} \nonumber \\
& \!\!\!\!\!\!+\!\! \sum\limits_{j \ne i}^{} {{p_{B,j}}\!\Big( \!{\mathbb{E}\Big(\! {{{\big| {\mathbf{\hat h}_{AR,i}^H\mathbf{h}_{BR,j}^{}} \big|}^2}} \!\Big) \!\!+\!\! \mathbb{E}\Big(\! {{{\big| {\mathbf{\hat h}_{BR,i}^H\mathbf{h}_{BR,j}^{}} \big|}^2}} \!\Big)} \!\Big)}.  \!\!
\vspace{-0.1cm}
\end{align}

The term $\mathbb{E}\big\{ {{{\big| {\mathbf{\hat h}_{AR,i}^H\mathbf{h}_{AR,j}^{}} \big|}^2}} \big\}$ can be expanded as
\begin{align}\label{ED1}
&\!\!\!\!\mathbb{E}\big\{ {{{\big| {\mathbf{\hat h}_{AR,i}^H\mathbf{h}_{AR,j}^{}} \big|}^2}} \big\}
\!\!= \!\!\mathbb{E}\big\{\! {{{\big| {\mathbf{\hat h}_{AR,i}^H\mathbf{\hat h}_{AR,j}^{}} \big|}^2}} \!\big\}\!\! + \!\! \mathbb{E}\big\{\! {{{\big| {\mathbf{\hat h}_{AR,i}^H\mathbf{e}_{AR,j}^{}} \big|}^2}} \!\big\} \nonumber \\
&\!\!\!\!\!\!\! + \mathbb{E}\!\!\left\{\! {\mathbf{\hat h}_{AR,i}^{}\mathbf{\hat h}_{AR,j}^H\mathbf{\hat h}_{AR,i}^H\mathbf{e}_{AR,j}^{}} \!\!\right\}\!\! + \!\!\mathbb{E}\!\!\left\{\! {\mathbf{\hat h}_{AR,i}^H\mathbf{\hat h}_{AR,j}^{}\mathbf{\hat h}_{AR,i}^{}\mathbf{e}_{AR,j}^H} \!\!\right\}\!.\!\!\!
\end{align}

According to Lemma 5 of \cite{zhang2014power}, an approximation of $\mathbb{E}\big\{ {{{\big| {\mathbf{\hat h}_{AR,i}^H\mathbf{h}_{AR,j}^{}} \big|}^2}} \big\}$ can be obtained as
\begin{align}\label{ED11}
&\mathbb{E}\left\{ {{{\left| {\mathbf{\hat h}_{AR,i}^H\mathbf{h}_{AR,j}^{}} \right|}^2}} \right\}
 \approx M\xi_{AR,ij}.
\end{align}

Similarly,  we can calculate the approximate expressions  of the remaining three terms $\mathbb{E}\big\{ {{{\big| {\mathbf{\hat h}_{BR,i}^H\mathbf{h}_{AR,j}^{}} \big|}^2}} \big\}$, $\mathbb{E}\big\{ {{{\big| {\mathbf{\hat h}_{AR,i}^H\mathbf{h}_{BR,j}^{}} \big|}^2}} \big\}$ and  $\mathbb{E}\big\{ {{{\big| {\mathbf{\hat h}_{BR,i}^H\mathbf{h}_{BR,j}^{}} \big|}^2}} \big\}$. Then, the approximate expression  of  $\mathbb{E}\left\{ { D_i^{}} \right\}$ can be obtained.

From \eqref{hat E_i^{IP}}, by using Lemma 5 of \cite{zhang2014power}, we can calculate $\mathbb{E}\left\{ {E_i^{}} \right\}$ as
\begin{align}\label{EEj}
\mathbb{E}\left\{ { E_i^{}} \right\}
&= \mathbb{E}\big\{ {{{\big\| {\mathbf{\hat h}_{AR,i}^{}} \big\|}^2}} \big\} + \mathbb{E}\big\{ {{{\big\| {\mathbf{\hat h}_{BR,i}^{}} \big\|}^2}} \big\} \nonumber \\
&=M(\omega_{AR,i}+\omega_{BR,i}).
\end{align}

By substituting the above results into \eqref{R_{1,i}^{IP}}  and \eqref{hat R_{XR,i}^{DF}}, we can respectively approximate ${R_{1,i}}$ and $R_{AR,i}$  as $\tilde R_{1,i}^{}$ in \eqref{tilde hat R_{1,i}^{DF}} and $\tilde R_{AR,i}^{}$  in  \eqref{tilde hat R_{AR,i}^{DF}} with $X=A$. Then, we use a similar method to obtain the approximation of $R_{BR,i}$ as $ \tilde R_{BR,i}^{}$ in \eqref{tilde hat R_{AR,i}^{DF}} with $X=B$. Thus, $R_{XR,i}$ in \eqref{hat R_{XR,i}^{DF}} can be approximated as $\tilde R_{XR,i}$ in \eqref{tilde hat R_{AR,i}^{DF}}.


Moreover, to calculate $R_{RX,i}$ in \eqref{R_{RX,i}^{IP}}, we will first caculate $R_{RA,i}$. The term $\mathbb{E}\big\{ {\mathbf{h}_{AR,i}^T\mathbf{\hat h}_{AR,i}^*} \big\}$ is given by
\begin{align}\label{Eg_{AR,i}^T}
&\mathbb{E}\big\{ {\mathbf{h}_{AR,i}^T\mathbf{\hat h}_{AR,i}^*} \big\} = M\omega_{AR,i}.
\end{align}

Then, we derive the term $Var\big\{ {\mathbf{h}_{AR,i}^T\mathbf{\hat h}_{AR,i}^*} \big\}$ as
\begin{align}\label{var1}
&\!\!Var\!\big\{\! {\mathbf{h}_{AR,i}^T\mathbf{\hat h}_{AR,i}^*} \!\big\} \!\nonumber \\
&= \frac{{M\beta _{AR,i}^2} \left[ {2{K_{AR,i}}{\eta _{AR,i}} + \eta _{AR,i}^2 + \frac{{{K_{AR,i}} + {\eta _{AR,i}}}}{{1 + {\tau _p}{p_p}\beta _{AR,i}^{}}}} \right]}{{{{\left( {{K_{AR,i}} + 1} \right)}^2}}}.
\end{align}

Similar to \eqref{var1}, the term $Var\big\{ {\mathbf{h}_{AR,i}^T\mathbf{\hat h}_{RB,i}^*} \big\}$ can be expressed as
\begin{align}\label{var2}
&\!\!Var\big\{ {\mathbf{h}_{AR,i}^T\mathbf{\hat h}_{BR,i}^*} \big\}
\!\approx M \xi_{BR,ii}.
\end{align}

Then, we derive the term $\sum\limits_{j \ne i}^{} {\big( {\mathbb{E}\big\{ {{{\big| {\mathbf{h}_{AR,i}^T\mathbf{\hat h}_{BR,j}^*} \big|}^2}} \big\} + \mathbb{E}\big\{ {{{\big| {\mathbf{h}_{AR,i}^T\mathbf{\hat h}_{AR,j}^*} \big|}^2}} \big\}} \big)}$ similarly. For $j \ne i$, we can obtain
\begin{equation}\label{Ej4}
\mathbb{E}\big\{ {{{\big| {\mathbf{h}_{AR,i}^T\mathbf{\hat h}_{BR,j}^*} \big|}^2}} \big\}
\approx M \xi_{BR,ij},
\end{equation}
\begin{align}\label{Ej5}
&\mathbb{E}\big\{ {{{\big| {\mathbf{h}_{AR,i}^T\mathbf{\hat h}_{AR,j}^*} \big|}^2}} \big\}
\approx M \xi_{AR,ij}.
\end{align}
\vspace{-0.2cm}

Finally, the term $\rho _{}^{}$ can be expressed as
\begin{align}\label{rhoj}
\rho =\sqrt {{{{p_r}}}/\Big({{M\!\sum\limits_{j = 1}^N \!{\left(\! {\omega_{AR,j} \!+\! \omega_{BR,j}} \! \right)} }}\! \Big) }.
\end{align}

Substituting \eqref{Eg_{AR,i}^T}-\eqref{rhoj} into \eqref{R_{RX,i}^{IP}} and \eqref{gamma{RX,i}^{DF}}, we can obtain  $\tilde R_{RA,i}^{}$ as \eqref{tilde hat R_{RA,i}^{DF}} with $X=A$.  Then, $R_{RB,i}$ can be approximated  by  \eqref{tilde hat R_{RA,i}^{DF}} with $X=B$ by using a similar method.
Thus, $R_{RX,i}$ in \eqref{R_{RX,i}^{IP}} can be approximated as $\tilde R_{RX,i}$ in \eqref{tilde hat R_{RA,i}^{DF}}.
Then, by substituting \eqref{tilde hat R_{AR,i}^{DF}} and \eqref{tilde hat R_{RA,i}^{DF}} into
\eqref{tilde hat R_{2,i}^{DF}}, we can obtain the expression of $\tilde R_{2,i}$.

Given the expressions of  $\tilde R_{1,i}$ and $\tilde R_{2,i}$, we complete the proof.

\end{appendices}

\bibliographystyle{IEEEtran}
\bibliography{IEEEabrv,DF}

\begin{thebibliography}{10}
\providecommand{\url}[1]{#1}
\csname url@samestyle\endcsname
\providecommand{\newblock}{\relax}
\providecommand{\bibinfo}[2]{#2}
\providecommand{\BIBentrySTDinterwordspacing}{\spaceskip=0pt\relax}
\providecommand{\BIBentryALTinterwordstretchfactor}{4}
\providecommand{\BIBentryALTinterwordspacing}{\spaceskip=\fontdimen2\font plus
\BIBentryALTinterwordstretchfactor\fontdimen3\font minus
  \fontdimen4\font\relax}
\providecommand{\BIBforeignlanguage}[2]{{%
\expandafter\ifx\csname l@#1\endcsname\relax
\typeout{** WARNING: IEEEtran.bst: No hyphenation pattern has been}%
\typeout{** loaded for the language `#1'. Using the pattern for}%
\typeout{** the default language instead.}%
\else
\language=\csname l@#1\endcsname
\fi
#2}}
\providecommand{\BIBdecl}{\relax}
\BIBdecl

\bibitem{Marzetta2010Noncooperative}
T.~L. Marzetta, ``Noncooperative cellular wireless with unlimited numbers of
  base station antennas,'' \emph{IEEE Trans. Wireless Commun.}, vol.~9, no.~11,
  pp. 3590--3600, Nov. 2010.

\bibitem{jin2014ergodic}
S.~Jin, X.~Liang, K.-K. Wong, X.~Gao, and Q.~Zhu, ``Ergodic rate analysis for
  multipair massive {MIMO} two-way relay networks,'' \emph{IEEE Trans. Wireless
  Commun.}, vol.~14, no.~3, pp. 1480--1491, Mar. 2015.

\bibitem{peng2013achievable}
Z.~Peng, W.~Xu, L.~Wang, and C.~Zhao, ``Achievable rate analysis and feedback
  design for multiuser {MIMO} relay with imperfect {CSI},'' \emph{IEEE Trans.
  Wireless Commun.}, vol.~13, no.~2, pp. 780--793, Feb. 2013.

\bibitem{suraweera2013multi}
H.~A. Suraweera, H.~Q. Ngo, T.~Q. Duong, C.~Yuen, and E.~G. Larsson,
  ``Multi-pair amplify-and-forward relaying with very large antenna arrays,''
  in \emph{Proc. IEEE Int. Conf. Commun. (ICC)}, Jun. 2013, pp. 4635--4640.

\bibitem{feng2017power}
J.~Feng, S.~Ma, G.~Yang, and B.~Xia, ``Power scaling of full-duplex two-way
  massive {MIMO} relay systems with correlated antennas and {MRC/MRT}
  processing,'' \emph{IEEE Trans. Wireless Commun.}, vol.~16, no.~7, pp.
  4738--4753, May 2017.

\bibitem{gao2013sum}
J.~Gao, S.~A. Vorobyov, H.~Jiang, J.~Zhang, and M.~Haardt, ``Sum-rate
  maximization with minimum power consumption for {MIMO} {DF} two-way relaying
  {Part I}: {Relay} optimization,'' \emph{IEEE Trans. Signal Process.},
  vol.~61, no.~14, pp. 3563--3577, Jul. 2013.

\bibitem{kong2018multipair}
C.~Kong, C.~Zhong, M.~Matthaiou, E.~Bj{\"o}rnson, and Z.~Zhang, ``Multipair
  two-way half-duplex {DF} relaying with massive arrays and imperfect {CSI},''
  \emph{IEEE Trans. Wireless Commun.}, vol.~17, no.~5, pp. 3269--3283, May
  2018.

\bibitem{sayeed2011continuous}
A.~M. Sayeed and N.~Behdad, ``Continuous aperture phased {MIMO}: A new
  architecture for optimum line-of-sight links,'' in \emph{Proc. IEEE Int.
  Symp. Ant. Propag. (APS)}, Jul. 2011, pp. 293--296.

\bibitem{zhang2014power}
Q.~Zhang, S.~Jin, K.-K. Wong, H.~Zhu, and M.~Matthaiou, ``Power scaling of
  uplink massive {MIMO} systems with arbitrary-rank channel means,'' \emph{IEEE
  J. Sel. Topics Signal Process.}, vol.~8, no.~5, pp. 966--981, Oct. 2014.

\bibitem{2018Multi}
X.~Li, M.~Matthaiou, Y.~Liu, H.~Q. Ngo, and L.~Li, ``Multi-pair two-way massive
  {MIMO} relaying with hardware impairments over rician fading channels,'' Dec.
  2018.

\bibitem{2019Power}
D.~Gu, J.~Yang, X.~Lei, and R.~Gao, ``Power scaling for multi-pair massive
  {MIMO} two-way relaying system under rician fading,'' \emph{Springer
  Telecommun. Syst.}, vol.~72, pp. 401--412, 2019.

\end{thebibliography}

\end{document}